\documentclass[11pt]{article}
\pdfoutput=1

\usepackage[margin=2.5cm]{geometry}

\usepackage[utf8]{inputenc}
\usepackage[T1]{fontenc}
\usepackage[colorlinks]{hyperref}
\usepackage{xcolor}

\usepackage{enumitem}
\usepackage{comment}

\hypersetup{
  linkcolor=[rgb]{0.3,0.3,0.6},
  citecolor=[rgb]{0.2, 0.6, 0.2},
  urlcolor=[rgb]{0.6, 0.2, 0.2}
}

\usepackage{mathtools, amsmath, amsthm, amsfonts, amssymb}

\usepackage{braket}

\theoremstyle{plain}
\newtheorem{theorem}{Theorem}
\newtheorem*{theorem*}{Theorem}
\newtheorem{proposition}[theorem]{Proposition}
\newtheorem{lemma}[theorem]{Lemma}
\newtheorem{corollary}[theorem]{Corollary}

\newtheorem{observation}[theorem]{Observation}

\newtheorem{definition}[theorem]{Definition}

\newtheorem*{problem*}{Problem}

\newtheorem{remark}[theorem]{Remark}
\newtheorem{example}[theorem]{Example}

\newcommand{\C}{\mathbb{C}}
\newcommand{\N}{\mathbb{N}}
\newcommand{\M}{\mathrm{M}}

\DeclareMathOperator{\rk}{rank}
\DeclareMathOperator{\rank}{rank}

\newcommand{\proj}[1]{\ket{#1}\!\bra{#1}}

\newcommand{\cD}{\mathcal{D}}

\newcommand{\cH}{\mathcal{H}}

\newcommand{\linspan}{\mathrm{span}}

\DeclareMathOperator{\diag}{diag}

\DeclareMathOperator{\Tr}{Tr}

\title{The Haemers bound of noncommutative graphs}
\author{Sander Gribling\thanks{Centrum Wiskunde \& Informatica (CWI) and QuSoft, Science Park 123, 1098XG, Amsterdam, Netherlands. {\tt \{gribling,Yinan.Li\}@cwi.nl}.} \and Yinan Li\footnotemark[1]}

\date{}

\begin{document}

\maketitle
\begin{abstract}
We continue the study of the quantum channel version of Shannon’s zero-error capacity problem.
We generalize the celebrated Haemers bound to noncommutative graphs (obtained from quantum channels). We prove basic properties of this bound, such as  additivity under the direct sum and submultiplicativity under the tensor product. The Haemers bound upper bounds the Shannon capacity of noncommutative graphs, and we show that it can outperform other known upper bounds, including noncommutative analogues of the Lov\'asz theta function (Duan--Severini--Winter, \emph{IEEE Trans.~Inform.~Theory}, 2013 and Boreland--Todorov--Winter, \emph{arXiv}, 2019).
\end{abstract}


\section{Introduction}

The celebrated \emph{Shannon capacity} of a graph $G$ is defined as 
\begin{equation}\label{eq: Shannon capacity of graph}
    \Theta(G)=\sup_k\sqrt[k]{\alpha(G^{\boxtimes k})}=\lim_{k\to\infty}\sqrt[k]{\alpha(G^{\boxtimes k})},
\end{equation}
where $\alpha(G)$ denotes the independence number of $G$ and $\boxtimes$ denotes the strong graph product~\cite{MR0089131}. The logarithm of $\Theta(G)$ characterizes the amount of information that can be transmitted through a classical communication channel, with zero error, where we allow an arbitrary number of uses of the channel and we measure the average amount of information transmitted per use of the channel. (The graph $G$ is the so-called \emph{confusability graph} associated to the channel, see Section~\ref{sec:classical to quantum}.) 
The Shannon capacity is not known to be \emph{computable}: Even though computing the independence number is NP-complete~\cite{MR0378476}, there exist graphs whose Shannon capacities are not achieved by taking the strong graph product with itself \emph{finitely} many times~\cite{54907}.
 
To upper bound the Shannon capacity, Lov\'asz introduced the celebrated theta function~\cite{lovasz1979shannon}, which can be cast as a semidefinite program and can be used to compute, e.g.,~$\Theta(C_5)$. Lov\'asz posed the question whether the Shannon capacity equals the theta function in general, which has been refuted by Haemers:  He introduced another upper bound on the Shannon capacity, now known as the Haemers bound, which can be strictly smaller than the theta function on some graphs~\cite{Haemers1978,haemers1979some}. 

Instead of a \emph{classical} communication channel, we could also consider a \emph{quantum} communication channel. Doing so leads to quantum information analogues of the aforementioned questions, the study of which was systematically initiated by Duan, Severini and Winter~\cite{duan2013}. {In Section~\ref{sec:classical to quantum} we show how the quantum setting generalizes the classical setting, which also motivates the following definitions. To (the Choi-Kraus representation of) a quantum channel $\Phi(A)=\sum_{k=1}^m E_kAE_k^\dagger$ ($\forall\ A\in M_n$) we associate the \emph{noncommutative (confusability) graph} $S_\Phi=\linspan\{E_k^\dagger E_k':~k,k'=1,\dots,m\}$.}
The noncommutative graph $S_\Phi$ completely characterizes the number of zero-error messages one can send through the quantum channel $\Phi$. 
More precisely, the independence number of $S\subseteq M_n$ is defined as the maximum number~$\ell$ for which there exist non-zero vectors (pure quantum states) $\ket{\psi_1},\dots,\ket{\psi_\ell}\in\C^n$ satisfying that $\bra{\psi_i}A\ket{\psi_j}=0$ for all distinct $i,j\in[\ell]$ and $A\in S$. We denote this by $\alpha(S)$. The Shannon capacity of a noncommutative graph $S$ is defined analogously as $\Theta(S)=\lim_{n\to\infty}\sqrt[n]{\alpha(S^{\otimes n})}$, where $\otimes$ is the tensor product~\cite{duan2013}. 

As in the classical setting, it is not known whether the Shannon capacity of noncommutative graphs is computable. We do know that computing the independence number of a noncommutative graph is QMA-hard~\cite{beigi2007complexity}.\footnote{In fact, computing the independence number of noncommutative graphs obtained from \emph{entanglement-breaking} channels is already QMA-complete~\cite{beigi2007complexity}.}
To upper bound the Shannon capacity of noncommutative graphs, it is natural to consider lifting bounds on the Shannon capacity of classical channels to the quantum setting. Duan, Severini and Winter introduced a quantum version of the Lov\'asz theta function on noncommutative graphs~\cite{duan2013}, which ``properly'' generalizes the theta function to the noncommutative graph setting. Recent studies have extended many other interesting graph notions to noncommutative graphs~\cite{stahlke2016,weaver2017quantum,8355678,weaver2018quantum,boreland2019sandwich}.
However, it remained an open question (as mentioned in~\cite{8355678}) how to generalize the Haemers bound to noncommutative graphs.

In this paper we show how to do so. We define a Haemers bound for noncommutative graphs, which canonically generalizes the classical Haemers bound of graphs (over complex numbers). Similar to the classical case, we prove that our bound upper bounds the Shannon capacity of noncommutative graphs. Our definition is inspired by the definition of noncommutative analogue the orthogonal rank~\cite{8355678}, combined with an 
observation that in the classical graph setting the orthogonal rank is a positive semidefinite version of the Haemers bound~\cite{peeters1996orthogonal}. We also compare our bound with other existing Shannon capacity bounds of noncommutative graphs. 

\section{Preliminaries}\label{sec: prelimiary}
Throughout, all scalars will be complex numbers. We use the Dirac notation for (unit) vectors, e.g., $\ket{i}$ stands for the $i$-th standard basis vector (whose dimension will be clear from context). Let $[n]=\{1,\dots,n\}$. Let $M_{n\times m}(\mathcal{X})$ be the set of $n\times m$ complex-valued matrices whose entries are from some ring $\mathcal{X}$. When $\mathcal{X}$ is omitted, we assume $\mathcal{X}=\C$. Let $M_n(\mathcal X) = M_{n \times n}(\mathcal{X})$. For a matrix $B\in M_{n\times m}(\mathcal X)$, we sometimes write $B=[B_{i,j}]_{i\in [n],j\in[m]}$, where $B_{i,j}\in \mathcal{X}$ denotes the $(i,j)$-th entry of $B$. A matrix $B\in M_{m \times n}$ has rank \emph{at most} $k$ if and only if there exist $C\in M_{k\times m}$ and $D\in M_{k\times n}$ such that $B=C^\dagger D$. We use $B\succeq 0$ to denote the positive-semidefiniteness of $B$ and $M_n^+$ to denote the set of \emph{positive semidefinite} matrices of size $n$-by-$n$. We use $\cD_n$ to denote the set (linear space) of diagonal matrices of size $n$-by-$n$. The \emph{trace} of a matrix $B\in M_n$ is the sum of its diagonal elements, i.e.,~$\Tr(B)=\sum_{i=1}^n B_{i,i}$. The \emph{Hilbert-Schmidt} inner product of two matrices $A,B\in M_{n\times n'}$ is $\Tr(A^\dagger B)$, where $A^\dagger$ denotes the conjugate transpose of $A$. The tensor product of two matrices $A\in M_{n\times n'}$ and $B\in M_{m\times m'}$ is the matrix
$$
A\otimes B=
\begin{bmatrix}
A_{1,1}B&\cdots&A_{1,n'}B\\
\vdots&&\vdots\\
A_{n,1}B&\cdots&A_{n,n'}B
\end{bmatrix}\in M_{n\times n'}(M_{m\times m'})\cong M_{nm\times n'm'}.
$$
The tensor product of two linear subspaces $S\subseteq M_{n\times n'}$ and $T\subseteq M_{m\times m'}$ is $S\otimes T=\linspan\{A\otimes B:\ A\in S,~B\in T\}$.

Throughout we assume that a graph is finite, simple, undirected and has no loops. Unless specified otherwise, we consider $n$-vertex graphs with vertices labeled by elements of the set $[n]$. For a graph $G=([n],E)$ we use $\overline{G}=([n],\Lambda_n\setminus E)$ to denote the complement graph of $G$, where $\Lambda_n=\{\{i,j\}:i,j\in[n]~\text{and}~i\neq j\}$. We use $K_n=([n],\Lambda_n)$ to denote the $n$-vertex complete graph and $\overline{K_n}=([n],\emptyset)$ to denote the $n$-vertex empty graph. 

An independent set of a graph $G=([n],E)$ is a subset of vertices $I\subseteq [n]$ such that any two vertices in $I$ are not adjacent, i.e., for all distinct $i,j\in I$ we have $\{i,j\}\not\in E$. The independence number of $G$, denoted by $\alpha(G)$, is the largest cardinality among all independent sets. For two graphs $G=([n],E)$ and $H=([n'],F)$, the strong graph product of $G$ and $H$, denoted by $G\boxtimes H$, is a graph with vertex set $[n]\times [n']$ and edge set $\{\{(i,k),(j,\ell)\}:~\{i,j\}\in E~\text{or}~i=j\in[n]~\text{and}~\{k,\ell\}\in F~\text{or}~k=\ell\in[n']\}\setminus\{\{(i,k),(i,k)\}:~i\in[n],~k\in[n']\}$. The independence number is supermultiplicative with respect to the strong graph product: $\alpha(G\boxtimes H)\geq \alpha(G)\alpha(H)$. The Shannon capacity of $G$, denoted by $\Theta(G)$, is defined as $\Theta(G)=\sup_k\sqrt[k]{\alpha(G^{\boxtimes k})}=\lim_{k\to\infty}\sqrt[k]{\alpha(G^{\boxtimes k})}$, where the limit equals the supremum due to Fekete's lemma.

We say there is a graph homomorphism from $G=([n],E)$ to $H=([n'],F)$, denoted by $G\to H$, if there exist a vertex map $\varphi: [n]\to[n']$ which preserves the adjacency relation, i.e.,~for any $\{i,j\}\in E$ we have $\{\varphi(i),\varphi(j)\}\in F$. 
We say that there is a graph cohomomorphism from $G$ to $H$, denoted by $G \leq H$, if $\overline G \to \overline H$.  The independence number can be equivalently defined in terms of graph (co)homomorphism: 
\[\alpha(G)=\max\{n:K_n\to\overline{G}\} =\max\{n:\overline{K_n}\leq G\}.
\]

\subsection{From the Shannon capacity of graphs to the Shannon capacity of noncommutative graphs} \label{sec:classical to quantum}
We briefly recall the development of noncommutative graph theory, which is originally motivated by the study of the quantum channel version of Shannon’s zero-error capacity problem~\cite{duan2013}.

The connection between information theory and graph theory was first observed by Shannon~\cite{MR0089131} in the study of the zero-error capacity problem: 
\begin{center}
    \emph{How many messages can be send through a communication channel with zero error?}
\end{center}
Here a classical communication channel $N$ is represented by the  transition probability function $N: X\to Y$ from the (finite) input alphabet $X$ to the (finite) output alphabet $Y$, where $N(y|x)$ denotes the probability of getting output $y$ conditioned on sending input $x$ through the channel $N$. Two input symbols $x,x' \in X$ can be confused by the receiver if there exists an output $y \in Y$ satisfying $N(y|x) > 0$ and $N(y|x') > 0$. To transmit messages through $N$ with zero error, these messages should be encoded into input symbols of which their outputs should not be confused by the receiver. 

To estimate the zero-error capacity, Shannon associated to each channel $N$ the \emph{confusability} graph $G_N$, where the vertices of $G_N$ are all possible input symbols in $X$, and two vertices $x,x'\in X$ are connected if they can be confused. A set of input symbols is said to be not confused if every pair of input symbols can not be confused. Such a set of input symbols thus forms an independent set in the confusability graph $G_N$. Hence, the maximum number of zero-error messages one can send via a single-use of $N$ equals the \emph{independence number} $\alpha(G_N)$ of $G_N$, which is a classical notion in graph theory. It is also not hard to see that to any graph $G$ one can also associate a classical communication channel $N$ such that $G$ and $G_N$ are the same (hence from now on we can focus on graphs instead of channels). Finally, if we are allowed to send length-$k$ codewords through $N$ (namely, use the channel $k$ times nonadaptively), the confusability graph of the resulting channel is given by $G_N^{\boxtimes k}$, the $k$-fold strong graph product of $G_N$. 
It is then natural to see the Shannon capacity of a graph (given in Equation~\eqref{eq: Shannon capacity of graph}) as the number of distinct messages per use some classical channel can communicate with no error, in the asymptotic limit. 

The same zero-error capacity problem can be also studied in the context of quantum information, where classical communication channels are replaced by quantum communication channels. (We refer the readers to~\cite{Nielsen2010} for a nice introduction.) Mathematically speaking, a \emph{quantum channel} is a \emph{completely positive and trace preserving} linear map $\Phi: M_n\to M_{n'}$. Equivalently, $\Phi :M_n \to M_{n'}$ is a map that admits a \emph{Choi-Kraus representation}: $\Phi(A)=\sum_{i=1}^m E_i A E_i^\dagger$ for all $A\in M_n$, where $\sum_{i=1}^m E_i^\dagger E_i=I_n$. The matrices $E_1,\dots,E_m\in M_{n'\times n}$ are called the \emph{Choi-Kraus operators} of $\Phi$. The input and output symbols of a quantum channel $\Phi: M_n\to M_{n'}$ are \emph{quantum states}, i.e.~positive semidefinite matrices with trace equal to one. Two quantum states are nonconfusable (perfectly distinguishable) if and only if they are orthogonal (with respect to the Hilbert-Schmidt inner product). Note that a classical communication channel $N: X\to Y$ can be viewed as the quantum channel $\Phi_N: M_{|X|}\to M_{|Y|}$, whose {Choi-Kraus} representation can be chosen as
\begin{equation}\label{eq: classical channel as quantum channel}
    \Phi_N(A)=\sum_{y\in Y,x\in X}N(y|x)\ket{y}\!\bra{x}A\ket{x}\!\bra{y}.
\end{equation}

To transmit (classical) zero-error messages through a quantum channel $\Phi: M_n\to M_{n'}$, one can encode the messages into quantum states and send those through the channel $\Phi$. In order to decode the messages with zero error, the channel-outputs corresponding to different messages should be orthogonal. The maximum number of zero-error messages one can send via a single-use of the quantum channel~$\Phi$ is the maximum number $\ell$ of distinct input states $\rho_1,\dots,\rho_\ell$ satisfying $\Tr(\Phi(\rho_i)^\dagger\Phi(\rho_j))=0$ for all $i,j\in[\ell]$. Due to the fact that quantum channels are completely positive, we may (w.l.o.g.) assume that each $\rho_i$ has rank $1$, that is, $\rho_i=\proj{\psi_i}$ for some unit vector $\ket{\psi_i}\in\C^n$. The orthogonality relations between the channel-outputs $\Phi(\proj{\psi_i})$ and $\Phi(\proj{\psi_j})$ then read as follows: 
$$\Tr(\Phi(\proj{\psi_i})^\dagger\Phi(\proj{\psi_j}))=\sum_{k,k'=1}^m\Tr(E_k\proj{\psi_i}E_k^\dagger E_{k'}\proj{\psi_j}E_{k'}^\dagger)=\sum_{k,k'=1}^m|\bra{\psi_i}E_k^\dagger E_{k'}\ket{\psi_j}|^2=0.$$
Observe that the maximum number of perfectly distinguishable messages, $\ell$, is completely determined by the subspace $S_\Phi=\linspan\{E_k^\dagger E_{k'}:k,k'\in[m]\}\subseteq M_n$. The subspace $S_\Phi$ is named as the \emph{noncommutative (confusability) graph} of the quantum channel $\Phi$ in~\cite{duan2013}. As we show below, the noncommutative graphs of quantum channels shall be viewed as a quantum generalization of the confusability graphs of classical channels. 
Note that $S_\Phi$ is self-adjoint, i.e.,  $A\in S_\Phi$ implies $A^\dagger \in S_\Phi$, and that $I_n\in S_\Phi$. Subspaces of $M_n$ satisfying these two properties are known as operator systems in functional analysis. In the rest of this paper, we refer to such subspaces as noncommutative graphs. This is justified, since we can also associate to every operator system a quantum channel of which the noncommutative (confusability) graph coincides with the operator system~\cite{duan2009super,cubitt2011}. 

As mentioned before, the independence number of a noncommutative graph $S\subseteq M_n$ (also denoted as $\alpha(S)$) is the maximum number $\ell$ of unit vectors $\ket{\psi_1},\dots,\ket{\psi_\ell}$ for which $\bra{\psi_i}A\ket{\psi_j}=0$ for all $i\neq j$ and $A\in S_\Phi$.\footnote{Another way to interpret this is to  view matrices in $S_\Phi$ as ``edges'' and all (unit) vectors in $\C^n$ as ``vertices''. The equalities $\bra{\psi_i}A\ket{\psi_j}=0$ for all $A \in S$ then indicate that $\ket{\psi_i}$ and $\ket{\psi_j}$ are ``nonadjacent''.} The independence number $\alpha(S)$ is exactly the maximum number of zero-error messages one can transmit via a single-use of any quantum channel $\Phi$ whose noncommutative graph is $S$. Note that the use of length-$k$ quantum codewords of a channel $\Psi$ results in the noncommutative graph $S_\Phi^{\otimes k}$. Define the Shannon capacity of a noncommutative graph $S$ as
\begin{equation}\label{eq: Shannon capacity of nc graph}
    \Theta(S)=\sup_k\sqrt[k]{\alpha(S^{\otimes k})}=\lim_{k\to\infty}\sqrt[k]{\alpha(S^{\otimes k})}.
\end{equation}
The Shannon capacity of a noncommutative graph is exactly the number of distinct messages per use some quantum channel can communicate with no error, in the asymptotic limit. 

Since classical channels are special cases of quantum channels. One can associate a classical channel $N$ (viewed as the quantum channel given in Equation~\eqref{eq: classical channel as quantum channel}) the noncommutative graph
\[
S_N=\linspan\{\ket{x}\!\bra{x'}:~\exists~y\in Y,~\sqrt{N(y|x)N(y|x')}>0\}.
\]
Note that $S_N$ is in one-to-one correspondence (up to graph isomorphism) with the confusability graph $G_N$ of $N$: it is the matrix space whose support equals $G_N$ (supplemented with the diagonal matrices). In other words, to every $n$-vertex graph $G = ([n],E)$ we associate the noncommutative graph 
\begin{equation}\label{eq: noncommutative graph of graph}
S_G \coloneqq \linspan\{\ket{i}\!\bra{j}:\ i=j\in [n]~{\rm or}~\{i,j\}\in E\} \subseteq M_n.
\end{equation}
We emphasize that this correspondence is canonical for several reasons. First of all, let us call two noncommutative graphs $S,T\subseteq M_n$ \emph{isomorphic} if there exist an $n\times n$ unitary matrix $U$ such that $U^\dagger S U=T$. (In the classical setting this corresponds to the situation where the input symbols are permuted.)
Then the graphs $G$ and $H$ are isomorphic if and only if $S_G$ and $S_H$ are isomorphic~\cite{ORTIZ2015128}. Second, the disjoint union of two graphs $G$ and $H$ is mapped to the \emph{direct sum} of the corresponding noncommutative graph $S_G$ and $S_H$ and the strong graph product of two graphs $G$ and $H$ is mapped to the \emph{tensor product} of the corresponding noncommutative graph $S_G$ and $S_H$. Last but not least, we have $\alpha(G) = \alpha(S_G)$~\cite{duan2013} and therefore, $\Theta(G) = \Theta(S_G)$. 

\subsection{Prior work on the Shannon capacities of graphs and noncommutative graphs}
\paragraph{The graph setting.} 
As we have mentioned before, it is not known how to compute the Shannon capacity of a graph. It is not even known whether it is computable at all. It therefore makes sense to consider bounds on the Shannon capacity. \emph{Lower bounds} often arise from explicit constructions of stable sets in some power of the (noncommutative) graph. For instance, the state-of-the-art lower bound on $\Theta(C_7)$ comes from an independent set of size $367$ in $C_7^{\boxtimes 5}$~\cite{polak2018new}. A natural way to obtain \emph{upper bounds} on the Shannon capacity is to obtain \emph{submultiplicative} (with respect to the strong graph product) upper bounds on the independence number. Examples of such upper bounds are the fractional clique-cover number~\cite{MR0089131}, the Lov\'asz theta number~\cite{lovasz1979shannon}, and parameters such as the (fractional) Haemers bound~\cite{Haemers1978,haemers1979some,blasiak2013graph,bukh2018fractional} and the orthogonal (projective) rank~\cite{lovasz1979shannon,manvcinska2016quantum}. The last three parameters are central to this work, which is why we define them below.

Let $G$ be a graph. The Lov\'asz theta number of $G$, denoted $\vartheta(G)$, is 
\begin{equation}\label{eq: theta function}
\vartheta(G)=\max\{\|I+T\|:~T_{i,j}=0~\text{if}~i=j\in[n]~\text{or}~\{i,j\}\in E,~I+T\succeq 0\};
\end{equation} 
the Haemers bound of $G$ (over $\C$), denoted $\cH(G)$, is 
\begin{equation}\label{eq: haemers bound}
\begin{split}
\cH(G)&=\min\{\rk(B):\ B_{i,i}=1~\forall~i\in[n],~B_{i,j}=0 \text{ if } \{i,j\}\not\in E(G)\}\\
&=\min\{\rk(B):\ B_{i,i}\neq 0~\forall~i\in[n],~B_{i,j}=0 \text{ if }\{i,j\}\not\in E(G)\};
\end{split}
\end{equation}
and finally the orthogonal rank\footnote{The name ``orthogonal rank'' (and its notation $\xi$) is normally used in the study of graph coloring, where adjacent vertices receive orthogonal vectors. In the study of independence number and Shannon capacity one usually considers $\overline \xi (G)$, the orthogonal rank of the complement of $G$, this equals the dimension of the vector space of an ``orthonormal representation'' (introduced by Lov\'asz in his original paper~\cite{lovasz1979shannon}), where \emph{nonadjacent} vertices receive orthogonal vectors. In this paper we are interested in the second setting and for brevity we will refer to $\overline \xi (G)$ as the  orthogonal rank of $G$. } of $G$ (over $\C$), denoted $\overline{\xi}(G)$, 
is 
\begin{equation}\label{eq: complement of orthogonal rank}
\overline{\xi}(G)=\min\{k:\ \exists\ket{\psi_1},\dots,\ket{\psi_n}\in\C^k,~\text{s.t.}~\langle\psi_i|\psi_i\rangle\neq 0~\text{and}~\langle\psi_i|\psi_{j}\rangle= 0~\text{if}~\{i,j\}\not\in E(G)\}. 
\end{equation}
At first glance the orthogonal rank of a graph $G$ seems unrelated to the Haemers bound, but we point out that the orthogonal rank can be viewed as a positive-semidefinite version of the Haemers bound (which has been mentioned implicitly in~\cite{peeters1996orthogonal}, and can also easily be obtained from~\cite{Hogben2017}).
\begin{observation}\label{obs: psd generalization}
For a graph $G=([n],E)$ we have
$$\overline{\xi}(G)=\min\{\rk(B):\ B_{i,i}=1~\forall~i\in[n],~B_{i,j}=0 \text{ if  }\{i,j\}\not\in E(G),~B\succeq 0\}.$$
\end{observation}
\begin{proof}
Indeed, the matrix $B = \mathrm{Gram}(\ket{\psi_1},\ldots,\ket{\psi_n}) =  \big[ \langle \psi_i | \psi_j\rangle \big]_{i,j \in [n]}$ is feasible and has rank at most $k$ if and only if the vectors $\ket{\psi_1},\ldots,\ket{\psi_n}$ can be taken in $\C^k$ and satisfy the orthogonality conditions of Equation~\eqref{eq: complement of orthogonal rank}.
\end{proof}

For any graph $G$, we have the following inequalities:
\begin{equation}\label{eq: graph inequality}
\alpha(G)\leq\Theta(G)\leq \vartheta(G)\leq\overline{\xi}(G),~~\alpha(G)\leq\Theta(G)\leq \cH(G)\leq\overline{\xi}(G).
\end{equation}
The Lov\'asz theta function and the Haemers bound are \emph{incomparable}: We have $\vartheta(C_5)=\sqrt{5}<3\leq \cH(C_5)$ and $\cH(G)\leq 7<9=\vartheta(G)$ when $G$ is taken as the complement of the \emph{Shl\"afli graph}~\cite{Haemers1978}.

\paragraph{The noncommutative graph setting.} 

Previously, work has been done on constructing noncommutative analogues of the Lov\'asz theta number~\cite{duan2013} and the orthogonal rank~\cite{stahlke2016,8355678}. 
Our goal is to provide a noncommutative analogue of Haemers bound and therefore we will go over the (very much related) noncommutative analogue of the orthogonal rank in more detail below.

In~\cite{stahlke2016}, the orthogonal rank of a {noncommutative} graph $S$, denoted $\overline{\xi}(S)$, is defined as
\begin{equation}\label{eq: quantum complement of orthogonal rank}
\overline{\xi}(S)=
\min\{k:\ \exists\ \text{ quantum channel } \Phi: M_n\to M_k,~\text{s.t.}~S_\Phi\subseteq S\}.
\end{equation}
The following proposition shows that $\overline \xi (S)$ is a proper noncommutative analogue of the orthogonal rank of graphs: 
\begin{proposition}[{\cite[Theorem~12]{stahlke2016}}]
Let $G=([n],E)$ be a graph and let $S_G$ be as in Equation~\eqref{eq: noncommutative graph of graph}. Then $\overline{\xi}(S_G)=\overline{\xi}(G)$.
\end{proposition}
\begin{proof}
We first show $\overline \xi(S_G) \leq \overline \xi(G)$. Let $\ket{\psi_1},\dots,\ket{\psi_n}\in\C^k$ be a feasible solution of $\overline{\xi}(G)$ (as in Equation~\eqref{eq: complement of orthogonal rank}) and let $\Phi: M_n\to M_k$ be the quantum channel defined as $\rho \mapsto \Phi(\rho):=\sum_{i=1}^n\ket{\psi_i}\!\bra{i}\rho\ket{i}\!\bra{\psi_i}$. We observe that $S_\Phi=\{\ket{i}\!\bra{j}:\braket{\psi_i|\psi_j}\neq 0\}$. Since $\braket{\psi_i|\psi_j}=0$ whenever $\{i,j\}\not\in E$, we therefore have $S_\Phi \subseteq S_G$, which implies that $\overline{\xi}(S_G)\leq \overline{\xi}(G)$. 

On the other hand, let $\Phi: M_n\to M_k$ be a quantum channel, with Choi-Kraus operators $\{E_1,\dots, E_m\}\subseteq M_{k\times n}$, that is a feasible solution of $\overline{\xi}(S_G)$. Since $\sum_{j=1}^m E_j^\dagger E_j= I_n$, for each $i\in [n]$ there exist at least one $j(i)\in[m]$ such that $E_{j(i)}\ket{i}$ is nonzero. Let $\ket{\psi_i}=E_{j(i)}\ket{i}\in \C^k$ for $i\in[n]$. We now show that $\langle \psi_i|\psi_{i'}\rangle =0$ for $\{i,i'\}\not\in E$. Note that $S_\Phi\subseteq S_G$ implies that for $\{i,i'\}\not\in E$ we have $S_\Phi \perp \ket{i'}\!\bra{i}$. Thus
$\langle \psi_i|\psi_{i'}\rangle=\bra{i}E_{j(i)}^\dagger E_{j(i')}\ket{i'}=\Tr(E_{j(i)}^\dagger E_{j(i')}\ket{i'}\!\bra{i})=0$ for all $\{i,i'\}\not\in E$. It follows that the vectors $\ket{\psi_1},\dots,\ket{\psi_n}$ are feasible for $\overline{\xi}(G)$ and therefore we have $\overline{\xi}(G)\leq \overline{\xi}(S_G)$.
\end{proof}

In~\cite{8355678}, an alternative definition of $\overline{\xi}$ for noncommutative graphs is formulated:
\begin{proposition}[Proposition $14$ in~\cite{8355678}]\label{prop: alternative formulation}
For a noncommutative graph $S\subseteq M_n$ we have
\begin{equation}\label{eq: rank formulation}
\overline{\xi}(S)=\min\{\rk(B):\ m\in\N,~B\in M_m(S)^+,~\sum_{i=1}^m B_{i,i}=I_n\}.
\end{equation}
\end{proposition}
\begin{proof}
Let $k$ and $\Phi: M_n\to M_k$ be a feasible solution of \eqref{eq: quantum complement of orthogonal rank}, where $\Phi: M_n\to M_k$ has Choi-Kraus operators $\{E_i:i\in[m]\}\subseteq M_{k\times n}$ satisfying $S_\Phi\subseteq S$. Define $B=\big[E_i^\dagger E_j\big]_{i,j\in[m]}\in M_m(S)$. Then $B$ is a feasible solution of \eqref{eq: rank formulation} since $\sum_{i=1}^m B_{i,i}=\sum_{i=1}^m E_i^\dagger E_i= I_n$ and $B$ can be written as $B=C^\dagger C$, where $C=\begin{bmatrix} E_1& \cdots &E_m\end{bmatrix}\in M_{k\times mn}$. Moreover, since $B=C^\dagger C$, we have $\rk(B)\leq k$.

On the other hand, let $B\in M_m(S)^+$ be a feasible solution of \eqref{eq: rank formulation} with $\rk(B)=k$. Let $B=C^\dagger C$ for $C=[C_1 \cdots C_m]$, where $C_i\in M_{k\times n}$. Define $\Phi: M_n\to M_k$ by $\Phi(A)=\sum_{i=1}^m C_iA C_i^\dagger$ for any $A\in M_n$. The map $\Phi: M_n\to M_k$ is a quantum channel since $\sum_{i=1}^m C_i^\dagger C_i= \sum_{i=1}^m B_{i,i} = I_n$. Finally, we have that $S_\Phi\subseteq S$ since $C_i^\dagger C_j = B_{i,j} \in S$ for $i,j\in[m]$. Thus $\Phi$ and $k$ form a feasible solution of~\eqref{eq: complement of orthogonal rank}.
\end{proof}

\noindent{\textbf{Remark:}~we have $\Theta(S) \leq \overline \xi (S)$ for all non-commutative graphs~\cite{stahlke2016,8355678}. 
Surprisingly, in~\cite{8355678} it was shown that there exist noncommutative graphs $S$ for which the Duan-Severini-Winter noncommutative analogue of $\vartheta$ (see~\cite{duan2013}) is strictly larger than $\overline \xi$. Recently, another noncommutative analogue of $\vartheta$ has been proposed in~\cite{boreland2019sandwich} from a geometric perspective, and it was shown that it always lies between $\Theta(S)$ and $\overline \xi(S)$.}

\section{Haemers bound for noncommutative graphs}
\subsection{Definition and consistency}
By comparing  the orthogonal rank (see Observation~\ref{obs: psd generalization}) to the definition of the Haemers bound for graphs (see Equation~\eqref{eq: haemers bound}), we see that both can be viewed as finding the smallest-rank matrices in a feasible region that is very similar: the only difference is that for $\overline{\xi}(G)$ the feasible matrices are additionally required to be positive semidefinite. 

Motivated by Proposition~\ref{prop: alternative formulation}, which gives a formulation of $\overline{\xi}$ for noncommutative graphs, we therefore define the Haemers bound for noncommutative graphs by dropping the positivity requirement on the feasible region: 
\begin{definition}[Haemers bound for noncommutative graphs]\label{def: quantum haemers bound}
Let $S\subseteq M_n$ be a noncommutative graph. The Haemers bound of $S$ (over $\C$) is defined as
\begin{equation}\label{eq: quantum haemers bound}
\cH(S)=\min\{\rk(B):\ {m\in\N},~B\in M_m(S),~\sum_{i=1}^m B_{i,i}=I_n\}.
\end{equation}
\end{definition}

We first show that this bound is a proper noncommutative analogue of Haemers bound.
\begin{proposition}
Let $G=([n],E)$ be a graph and let $S_G$ be defined as in \eqref{eq: noncommutative graph of graph}. Then $\cH(S_G)=\cH(G)$.
\end{proposition}
\begin{proof}

We first show that $\cH(S_G)\leq \cH(G)$. Let $B$ be a feasible solution of \eqref{eq: haemers bound} with $\rk(B)=k$ and $B_{i,i}=1$ for $i\in[n]$. Decompose $B$ as $B=C^\dagger D$ where $C,D\in M_{k\times n}$. Denote the (normalized) columns of $C$ and $D$ as $\{\ket{C_1},\dots,\ket{C_n}\}$ and $\{\ket{D_1},\dots,\ket{D_n}\}$ respectively. Then $B=\begin{bmatrix}\langle C_i|D_j\rangle\end{bmatrix}_{i,j\in [n]}$. Define the matrix $B'=\begin{bmatrix}\langle C_i|D_j\rangle\ket{i}\!\bra{j}\end{bmatrix}_{i,j\in [n]}$. We show that $n$ and $B'$ is a feasible solution of~\eqref{eq: quantum haemers bound}. First, note that $\langle C_i|D_j\rangle\ket{i}\!\bra{j}\in S_G$ for $i,j\in[n]$, since $\langle C_i|D_j\rangle = B_{i,j} =0$ when $\{i,j\}\not\in E$. Second,  we have $\sum_{i=1}^n B'_{i,i}=\sum_{i=1}^n \langle C_i|D_i\rangle\ket{i}\!\bra{i}=\sum_{i=1}^n \ket{i}\!\bra{i}=I_n$. To bound the rank of $B'$, note that $B'=(C')^\dagger D'$, where $C'=\begin{bmatrix}\ket{C_1}\!\bra{1}&\cdots&\ket{C_n}\!\bra{n}\end{bmatrix}$ and $D'=\begin{bmatrix}\ket{D_1}\!\bra{1}&\cdots&\ket{D_n}\!\bra{n}\end{bmatrix}$. Since $C',D'\in M_{k\times n^2}$, we have $\rk(B')\leq k$ and therefore $\cH(S_G)\leq \cH(G)$.

We then show that $\cH(G)\leq \cH(S_G)$. Let $m$ and $B$ be a feasible solution of \eqref{eq: quantum haemers bound} with $\rk(B)=k$. Write $B=C^\dagger D$ for $C,D\in M_{k\times mn}$. We group the columns of $C$ and $D$ according to the block-structure of $B$: let $C=\begin{bmatrix} C_1 & \cdots & C_m \end{bmatrix}$ and $D=\begin{bmatrix} D_1 & \cdots & D_m\end{bmatrix}$ where $C_i,D_{i'}\in M_{k\times n}$ for $i,i'\in [n]$. By the feasibility of $B$ we have that $C_i^\dagger D_{i'}\in S_G$ for all $i,i'\in[n]$ and $\sum_{j=1}^m C_j^\dagger D_j=I_n$. By the second condition, we know that for each $i\in [n]$, we can pick a $j(i)\in[m]$ such that $\bra{i}C_{j(i)}^\dagger D_{j(i)}\ket{i}\neq 0$ (if there is more than one such index, pick an arbitrary one). Let $B'=\begin{bmatrix}\bra{i}C_{j(i)}^\dagger D_{j(i')}\ket{i'}\end{bmatrix}_{i,i'\in[n]}$. We will show that $B'$ is a feasible solution to~\eqref{eq: haemers bound} (the second formulation). It is easy to see that the diagonal entries of $B'$ are nonzero, thus we only need to show that $B'_{i,i'}=\bra{i}C_{j(i)}^\dagger D_{j(i')}\ket{i'}=0$ for $\{i,i'\}\not\in E$. This follows from $C_{j(i)}^\dagger D_{j(i')}\in S_G$ and $\ket{i}\!\bra{i'}\perp S_G$ if $\{i,i'\}\not\in E$. Finally, to bound the rank of $B'$, note that $B'$ can be written as $U^\dagger V$, where $U=\begin{bmatrix} C_{j(1)}\ket{1} & \cdots & C_{j(n)}\ket{n} \end{bmatrix}\in M_{k\times n}$ and $V=\begin{bmatrix} D_{j(1)}\ket{1} & \cdots & D_{j(n)}\ket{n} \end{bmatrix}\in M_{k\times n}$. We thus have $\rk(B')\leq k$ and it follows that $\cH(G)\leq \cH(S_G)$. 
\end{proof}

\subsection{Upper bound on the Shannon capacity of noncommutative graphs}
We first show that the noncommutative analogue of the Haemers bound is submultiplicative with respect to the tensor product:
\begin{proposition}\label{prop: submultiplicative}
Let $S\subseteq M_n$ and $T\subseteq M_{n'}$ be noncommutative graphs, we have 
\[
\cH(S\otimes T)\leq\cH(S)\cH(T).
\]
\end{proposition}
\begin{proof}
Let $B_1$ and $m_1$ be a feasible solution of $\cH(S)$ and $B_2$ and $m_2$ be a feasible solution of $\cH(T)$. We construct a feasible solution of $S\otimes T$. Let $B=B_1\otimes B_2$ and $m=m_1m_2$. It is easy to see that $B_1\otimes B_2\in M_m(S\otimes T)$. Moreover,
$$\sum_{i=1}^m B_{i,i}=\sum_{i_1=1}^{m_1}\sum_{i_2=1}^{m_2}(B_1)_{i_1,i_1}\otimes (B_2)_{i_2,i_2}=I_{m_1}\otimes I_{m_2}=I_m.$$
Since $\rk(B)=\rk(B_1\otimes B_2)=\rk(B_1)\rk(B_2)$, we conclude that $\cH(S\otimes T)\leq\cH(S)\cH(T)$.
\end{proof}
\begin{remark}
The inequality in Proposition~\ref{prop: submultiplicative} can be strict; it was shown by Bukh and Cox that $\cH(C_5^{\boxtimes 2})\leq 8 < 9\leq \cH(C_5)^2$, see~\cite[Proposition~9]{bukh2018fractional}.
\end{remark}

We next show that the Haemers bound of noncommutative graphs upper bounds the independence number of noncommutative graphs and therefore also its Shannon capacity: 
\begin{theorem}\label{prop: upper bound independence number}
Let $S\subseteq M_n$ be a noncommutative graph. We have $\alpha(S)\leq\Theta(S)\leq\cH(S)\leq\overline{\xi}(S)$.
\end{theorem}
\begin{proof}
The last inequality holds since $\cH(S)$ (Equation~\eqref{eq: quantum haemers bound}) is a relaxation of $\overline{\xi}(S)$ (Equation~\eqref{eq: rank formulation}).

To see the first two inequalities, we only need to show $\alpha(S)\leq\cH(S)$ since by the previous proposition $\cH(S)$ is submultiplicative. 

Let $\alpha(S)=\ell$ and let $\ket{\psi_1},\dots,\ket{\psi_\ell}\in\C^n$ satisfy $\bra{\psi_i}A\ket{\psi_j}=0$ for all $i\neq j$ and $A\in S$. Let $B$ and $m$ be a feasible solution of $\cH(S)$. We show that $\rk(B)\geq \ell$. Let $U=\begin{bmatrix}\ket{\psi_1} & \cdots & \ket{\psi_\ell} \end{bmatrix}\in M_{n\times \ell}$ and write $B=\begin{bmatrix}B_{i,j}\end{bmatrix}_{i,j\in[m]}$ where $B_{i,j}\in S$ for all $i,j\in[m]$. Note that 
\[
\rk(B)\geq \rk((I_m\otimes U^\dagger) B(I_m\otimes U)).
\]
Let $D_{i,j}:=U^\dagger B_{i,j} U$ and note that $D_{i,j}=\diag(\bra{\psi_1}B_{i,j}\ket{\psi_1},\cdots,\bra{\psi_\ell}B_{i,j}\ket{\psi_\ell})\in M_\ell$. We set $D:=\begin{bmatrix}D_{i,j}\end{bmatrix}_{i,j\in [m]} = (I\otimes U^\dagger) B(I\otimes U) \in M_{m\ell}$. We claim that $\rk(D)\geq \ell$. To see this we use that $\sum_{j=1}^m B_{j,j}=I_n$. For any $i\in[\ell]$, there exists at least one $j(i)\in[m]$, such that $\bra{\psi_i}B_{j(i),j(i)}\ket{\psi_i}\neq 0$. Then the submatrix of $D$ consisting of the $(j(i),i)$-th row and column for all $i\in [\ell]$ is diagonal with every diagonal entry being nonzero.
We conclude that $\rk(B)\geq \rk(D)\geq \ell$.
\end{proof}

All inequalities in the above theorem can be strict. This follows from the fact that they can be strict for graphs. Similarly, we point out that the Haemers bound of noncommutative graphs is incomparable with existing noncommutative analogues of the theta function ($\vartheta$ and $\tilde{\vartheta}$ introduced in~\cite{duan2013} and recently $\theta$ and $\hat{\theta}$ in~\cite{boreland2019sandwich}). This again follows from the fact that $\vartheta$ and $\cH$ are incomparable for graphs.

\subsection{Properties of the Haemers bound and examples}

We first give an upper bound on the size of the block-matrix needed in the definition of the Haemers bound of noncommutative graphs. The result is similar to Proposition IV.7 from~\cite{8355678}, where they show that for $\overline{\xi}(S)$ we may restrict to $m \leq 2n^3$ in Equation~\eqref{eq: rank formulation}. 
\begin{proposition}\label{prop: upper bound on m}
Let $S\subseteq M_n$, the optimal solution of $\cH(S)$ can be achieved with $m\leq n^{4}$. 
\end{proposition}
\begin{proof}
First, let us note that $\cH(S) \leq \overline{\xi}(S) \leq n$ since $B = I_n \in M_1(S)$ is a feasible solution of rank~$n$. 

Let $m \in \N$ and $B \in M_m(S)$ be feasible for $\cH(S)$ with $\rk(B)=k \leq n$. Let us write $B = C^\dagger D$ with $C,D \in M_{k \times mn}$. Say $C = \begin{bmatrix}C_1 & \cdots & C_m\end{bmatrix}$ and $D = \begin{bmatrix}D_1 & \cdots & D_m\end{bmatrix}$ where $C_i,D_i \in M_{k\times n}$. Then, feasibility of $B$ implies that $\sum_{i=1}^m C_i^\dagger D_i = I_n$. The crucial observation is now that there can be at most $kn$ linearly independent matrices $C_i$ (likewise for the matrices $D_i$). It follows that $I_n = \sum_{i=1}^m C_i^\dagger D_i = \sum_{i \in I, j \in J}\alpha_{i,j} C_i^\dagger D_j$ for some index sets $I,J \subseteq [m]$ with $|I|,|J|\leq kn$ and coefficients $\alpha_{i,j} \in \C$ (for $i \in I, j \in J$). 
We will now construct a matrix $B' \in M_{(kn)^2}(S)$ which is feasible for $\cH(S)$ and has rank at most $k$. 
For $i \in I$ set 
\begin{align*}
\bar C_{i} &:= \big[\underbrace{C_{i} \ \cdots \ C_{i}}_{|J| \text{ times}}\big] \in M_{k \times |J|n}, \qquad 
\bar D_i &:= \begin{bmatrix}(\alpha_{i,j_1} D_{j_1})\ (\alpha_{i,j_2} D_{j_2}) \ \cdots \ (\alpha_{i,j_{|J|}} D_{j_{|J|}})\end{bmatrix} \in M_{k \times |J|n},
\end{align*}
where $j_1,\ldots, j_{|J|}$ are the elements of $J$ (for later use, let similarly $I= \{i_1,\ldots, i_{|I|}\}$). Next define 
\begin{align*}
C' &= \big[\bar C_{i_1} \ \bar C_{i_2} \ \ \cdots \ \ \bar C_{i_{|I|}} \big] \in M_{k \times |I||J|n} \\
D' &= \big[\bar D_{i_1} \ \bar D_{i_2} \ \ \cdots \ \ \bar D_{i_{|I|}} \big] \in M_{k \times |I||J|n}.
\end{align*}
Finally, set $B' = (C')^\dagger D'$ and observe that $B' \in M_{|I||J|}(S)$, that the sum of the diagonal blocks of $B'$ is $\sum_{i \in I, j \in J}\alpha_{i,j} C_i^\dagger D_j = I_n$, and that $\rk(B')\leq k$. 

To conclude the proof it suffices to note that $|I|,|J| \leq kn$ and therefore we may restrict our attention to $m \leq (kn)^2 \leq n^4$ in the definition of $\cH(S)$ (see Equation~\eqref{eq: quantum haemers bound}). 
\end{proof}
The above proposition implies that the Haemers bound is computable.
\begin{corollary}\label{cor: computability}
The Haemers bound of noncommutative graphs is computable.
\end{corollary}
\begin{proof}
Let $S \subseteq M_n$ be a noncommutative graph. Proposition~\ref{prop: upper bound on m} tells us that in Definition~\ref{def: quantum haemers bound} we may restrict our attention to $m \leq n^{4}$, i.e., matrices $B$ of size polynomial in $n$. In the proof of the previous proposition we have seen that $\cH(S) \leq n$, therefore it is an integer between $1$ and $n$ and we may compute it by solving several feasibility problems. Each feasibility problem asks whether there exists a matrix $B \in M_{n^4}(S)$ whose diagonal blocks sum up to the identity matrix has rank at most $k$ for some $k \in [n]$. Such a problem can be viewed as asking whether or not a system of polynomial equations has a common root (in $\C$): The condition that $\rk(B) \leq k$ is equivalent to the condition that all the $(k+1)\times (k+1)$-minors of $B$ are equal to zero. 

\emph{Hilbert's Nullstellensatz} implies that a system of polynomials has a common root if and only if $1$ does not belong to the ideal generated by those polynomials. The latter can be tested using \emph{Gr\"obner bases}. We refer to, for instance,~\cite{CLO15} for more details. 
\end{proof}

We now give a formulation of $\cH(S)$ that is similar to Stahlke's definition of $\overline \xi(S)$ (see Equation~\eqref{eq: quantum complement of orthogonal rank}). For this, recall that a general trace-preserving linear map $\Psi: M_n\to M_{n'}$ can be written as $\Psi(X)=\sum_{i=1}^m E_iXF_i^\dagger$ for all $X\in M_n$, where $E_1,\dots,E_m,F_1,\dots,F_m\in M_{n'\times n}$ satisfy $\sum_{i=1}^m F_i^\dagger E_i=I_n$, see~\cite[Theorem 2.26]{watrous_2018}. Let $T_\Psi=\linspan\{F_i^\dagger E_j:~i,j\in[m]\}$. We have the following:
\begin{proposition}
Let $S\subseteq M_n$, we have 
\begin{equation}\label{eq: operation interpretation of Haemers bound}
    \cH(S)=\min\{k:\ \exists\ \text{ trace-preserving linear map } \Psi: M_n\to M_k\ \text{s.t.}\ T_\Psi\subseteq S\}.
\end{equation}
\end{proposition}
\begin{proof}
To show ``$\geq$'', let $m \in \N$ and $B \in M_m(S)$ be feasible for $\cH(S)$ in Equation~\eqref{eq: quantum haemers bound} with $\rk(B)=k$. Let us write $B = C^\dagger D$ with $C,D \in M_{k \times mn}$. Say $C = \begin{bmatrix}C_1 & \cdots & C_m\end{bmatrix}$ and $D = \begin{bmatrix}D_1 \cdots D_m\end{bmatrix}$ where $C_i,D_i \in M_{k\times n}$. Then, feasibility of $B$ implies that $\sum_{i=1}^m C_i^\dagger D_i = I_n$. Let $\Psi: M_n\to M_k$ be defined as $\Psi(A)=\sum_{i=1}^m D_iAC_i^\dagger$ for any $A\in M_n$. Then $\Psi$ is trace-preserving and $T_\Psi=\linspan\{C_i^\dagger D_j:~i,j\in[m]\}\subseteq S$. 

Conversely, to show ``$\leq$'', let $\Psi: M_n\to M_k$ be a feasible solution of the right-hand side of Equation~\eqref{eq: operation interpretation of Haemers bound}, and let $E_1,\dots, E_m,F_1,\dots,F_m\in M_{k\times n}$ be such that $\Psi(A)=\sum_{i=1}^m E_i A F_i^\dagger$ for all $A\in M_n$ and $\sum_{i=1}^m F_i^\dagger E_i=I_n$. Define the matrix $B=\big[ F_i^\dagger E_j \big]_{i,j\in[m]}$. Then $B$ and $m$ is a feasible solution of $\cH(S)$ with $\rk(B)\leq k$.
\end{proof}

With this characterization we show that the Haemers bound is monotone with respect to noncommutative graph \emph{cohomomorphism}, a notion that was introduced in~\cite{stahlke2016} (see also~\cite{Li2018quantum}). Let $S\subseteq M_n$ and $T\subseteq M_{n'}$ be two noncommutative graphs. We say there is a \emph{cohomomorphism} from $S$ to $T$, denoted as $S\leq T$, if there exists a quantum channel $\Phi: M_n\to M_{n'}$ with Choi-Kraus operators $E_1,\dots,E_m\in M_{n\times n'}$, such that for every $B\in T$ and $i,j\in[m]$ we have $E_i^\dagger B E_j\in S$. It is called cohomomorphism because of its interpretation for graphs: for graphs $G$ and $H$ we have $S_G\leq S_H$ if and only if there is a homomorphism from $\overline{G}$ to $\overline{H}$~\cite{stahlke2016,Li2018quantum}.
\begin{proposition}\label{eq: monotone}
For noncommutative graphs $S\subseteq M_n$ and $T\subseteq M_{n'}$, $S\leq T$ implies $\cH(S)\leq\cH(T)$.
\end{proposition}
\begin{proof}
Let $\Psi: M_{n'}\to M_k$ be a trace-preserving linear map acting as $\Psi(A)=\sum_{i=1}^m E_iAF_i^\dagger$ for any $A\in M_{n'}$, where $E_1,\dots,E_m,F_1,\dots,F_m\in M_{k\times n'}$. And let $\Psi$ be feasible to $\cH(T)$ as in Equation~\eqref{eq: operation interpretation of Haemers bound}. Let $\Phi: M_n\to M_{n'}$ be a quantum channel with Choi-Kraus operators $D_1,\dots,D_{m'}\in M_{n'\times n}$, such that for any $A\in T$ and $i,j\in[m']$, $D_i^\dagger AD_j\in S$. We claim that for the linear map $\Psi': M_n\to M_k$ acting as $\Psi'(A)=\Psi(\Phi(A))$ for any $A\in M_n$, $\Psi'$ is trace preserving and $T_{\Psi'}=\linspan\{D_{j}^\dagger F_{i}^\dagger E_{i'}D_{j'}:~i,i'\in[m],~j,j'\in[m']\}\subseteq S$. $\Psi'$ being trace preserving is easy to see since $\Phi$ and $\Psi$ are trace-preserving. Since for every $i,i'\in[m]$, $F_{i}^\dagger E_{i'}\in T_\Psi\subseteq T$. We have $D_{j}^\dagger F_{i}^\dagger E_{i'}D_{j'}\in S$ for any $j,j'\in[m']$ and $i,i'\in[m]$. Thus $\Psi'$ is a feasible solution of $\cH(S)$ as in Equation~\eqref{eq: operation interpretation of Haemers bound}, which implies $\cH(S)\leq \cH(T)$.
\end{proof}

The above monotonicity result allows us to give an alternative proof of Proposition~\ref{prop: upper bound independence number}: Note that $\alpha(S)=\max\{\ell:~\cD_\ell\leq S\}$~\cite{stahlke2016} (See also~\cite[Lemma 14]{Li2018quantum}). Letting $\alpha(S)=\ell$, we obtain $\cH(S) \geq \cH(\cD_\ell) = \ell$.

\medskip

The following proposition lists some other basic properties of the bound $\cH(S)$. 
\begin{proposition} \label{prop: properties}
Let $S\subseteq M_n$ and $T\subseteq M_{n'}$ be noncommutative graphs. The following holds:
\begin{itemize}
\item[(1)]  For any $n\times n$ unitary matrix $U$ we have $\cH(S)=\cH(U^\dagger S U)$.
\item[(2)]  $\cH(S\oplus T)= \cH(S)+\cH(T)$.
\end{itemize}
\end{proposition}
\begin{proof}
(1)~It suffices to show $\cH(U^\dagger S U) \leq \cH(S)$. Let $m$, $B$ be a feasible solution of $\cH(S)$, then $m$, $B'=(I_m\otimes U^\dagger)B(I_m\otimes U)$ is a feasible solution of $\cH(U^\dagger S U)$ and $\rk(B')= \rk(B)$, therefore $\cH(U^\dagger S U) \leq \cH(S)$. 

(2)~We first show that $\cH(S\oplus T)\leq \cH(S)+\cH(T)$. Let $B_1$ and $m_1$ be a feasible solution of $\cH(S)$ and $B_2$ and $m_2$ be a feasible solution of $\cH(T)$. Without loss of generality we assume $m_1=m_2=m$. Let $B_1=\begin{bmatrix}X_{i,j}\end{bmatrix}_{i,j\in[m]}$ and $B_2=\begin{bmatrix}Y_{i,j}\end{bmatrix}_{i,j\in[m]}$ with $X_{i,j}\in S$ and $Y_{i,j}\in T$. The matrix \[
B=\begin{bmatrix}\begin{bmatrix}X_{i,j}&0\\0&Y_{i,j}\end{bmatrix}\end{bmatrix}_{i,j\in [m]}\] and $m$ is a feasible solution of $\cH(S\oplus T)$ with $\rk(B)=\rk(B_1)+\rk(B_2)$.

We now show that $\cH(S\oplus T)\geq \cH(S)+\cH(T)$. Let $B$ and $m$ be a feasible solution of $\cH(S\oplus T)$. 
We have $M_m(S\oplus T) \simeq M_m(S) \oplus M_m(T)$ where the isomorphism is given by a permutation, let us denote the permutation with $\pi$. Then $\pi B \pi^\dagger = B_1 \oplus B_2$ where $B_1 \in M_m(S)$ and $B_2 \in M_m(T)$ are feasible for $\cH(S)$ and $\cH(T)$ respectively. It remains to observe that 
\[
    \rank(B) = \rank(B_1 \oplus B_2) = \rank(B_1)+\rank(B_2) \geq \cH(S) + \cH(T). \qedhere
\]
\end{proof}

Let us now compute the value of $\cH(S)$ for some basic noncommutative graphs $S$. 

\begin{example}
Let $S = \C I_n$ be the subspace containing only scalar multiples of the $n\times n$ identity matrix, then $\cH(S)=n$. 

Indeed, any feasible solution $B$ of $\cH(\C I_n)$ has at least one non-zero diagonal block $B_{i,i} \in S=\C I_n$. This diagonal block provides a submatrix of $B$ with rank $n$. Therefore $\cH(\C I_n) \geq n$. It is also easy to see that $B=I_n \in M_1(\C I_n)$ provides a feasible solution with rank exactly $n$. 
\end{example}
\begin{example} Let $S=\cD_n$ be the subspace of $M_n$ containing all diagonal matrices, then, by part~(2) of Proposition~\ref{prop: properties}, we have $\cH(\cD_n) = n$.
\end{example}

\begin{lemma} \label{lem: equal to one}
Let $S \subseteq M_n$ be a noncommutative graph. Then $\cH(S) =1$ if and only if $S = M_n$.
\end{lemma}
\begin{proof}
Suppose $S = M_n$ and set $u = \oplus_{i=1}^n \ket{i}$. Then $B = u u^\dagger$ is feasible for $\cH(M_n)$ and has rank equal to $1$. Therefore $\cH(M_n)=1$. 

For the other direction, let $S$ be a noncommutative graph for which $\cH(S)=1$. Let $B \in \M_m(S)$ be a feasible solution to $\cH(S)$ with rank equal to $1$. Say $B = uv^\dagger$ where $u,v \in \C^{mn}$. Decompose $u$ as $u=\oplus_{i=1}^m \ket{\psi_i}$ where $\ket{\psi_i} \in \C^n$ for each $i \in [m]$. Similarly, let $v = \oplus_{i=1}^m \ket{\phi_i}$. From the feasibility of $B$ it follows that $\sum_{i=1}^m \ket{\psi_i}\!\bra{\phi_i} = I_n$. In particular this implies that 
\[
\linspan\{\ket{\psi_1},\ldots, \ket{\psi_m}\} = \C^n = \linspan\{\ket{\phi_1},\ldots,\ket{\phi_m}\}.
\]
This in turn implies that $\linspan\{\ket{\psi_i}\!\bra{\phi_j}: i,j \in [m]\} = M_n$. At the same time, since $B \in \M_m(S)$, we have that $\ket{\psi_i}\!\bra{\phi_j} \in S$ for each $i,j\in [m]$. Therefore $S=M_n$. 
\end{proof}
\begin{example}
Let $S_n=\linspan\{I_n,\ket{i}\!\bra{j}:\ i\neq j\in[n]\}\subseteq M_n$. It follows from the above Lemma~\ref{lem: equal to one} that $\cH(S_n) \geq 2$ whenever $n \geq 2$. When $n=2$, i.e., for $S_2 = \left\{ \begin{pmatrix} a & b \\ c & a\end{pmatrix}: a,b,c \in \C\right\}$, this lower bound is tight. Indeed, $B = I_2 \in M_1(S_2)$ is feasible for $\cH(S_2)$ and has rank exactly equal to $2$. 

For $k\geq 2$, we have $\cH(S_k\otimes S_{k^2})\leq\overline{\xi}(S_k\otimes S_{k^2})\leq k^2<k^3\leq\vartheta(S_k\otimes S_{k^2})\leq\tilde{\vartheta}(S_k\otimes S_{k^2})$ where all but the first inequality were shown in~\cite{8355678}. Thus, the ratio $\cH(S)/\vartheta(S)$ (and $\cH(S)/\tilde{\vartheta}(S)$) can be arbitrarily small.
\end{example}

\begin{example}
In~\cite{Wang2018}, they presented a family of noncommutative graphs $S_\gamma$ (with parameter $\gamma$) where the Duan-Severini-Winter noncommutative analogue of $\vartheta$,  $\tilde{\vartheta}(S_\gamma)$ can be strictly larger than the entanglement-assisted Shannon capacity. Explicitly, 
\[
S_\gamma=\linspan\{\ket{1}\!\bra{3},\ket{3}\!\bra{1},\sin^2\gamma\proj{2}+\proj{3},\cos^2\gamma\proj{2}+\proj{1}\}\subseteq M_3,
\]
and $\tilde \vartheta(S_\gamma)=2 + \cos^2\gamma +\cos^{-2}\gamma\geq 4$ when $\gamma\in[0,\pi/2)$. On the other hand, note that $\cH(S_\gamma)\leq 3$ for every $\gamma$, which is strictly smaller than $\tilde \vartheta(S_\gamma)$. 
\end{example}

\section*{Acknowledgement}
We thank Monique Laurent for helpful discussions, Andreas Winter for pointing out the computability of our quantum Haemers bound (Corollary~\ref{cor: computability}) through Hilbert’s Nullstellensatz, and Christian Majenz for pointing out the reference~\cite{watrous_2018} related to trace-preserving linear maps.

\raggedright
\bibliographystyle{alpha}
\bibliography{all}

\begin{thebibliography}{HPRS17}

\bibitem[BC19]{bukh2018fractional}
B.~{Bukh} and C.~{Cox}.
\newblock {On a Fractional Version of Haemers’ Bound}.
\newblock {\em IEEE Transactions on Information Theory}, 65(6):3340--3348, June
  2019.

\bibitem[Bla13]{blasiak2013graph}
Anna Blasiak.
\newblock {\em {A Graph-theoretic Approach to Network Coding}}.
\newblock PhD thesis, Cornell University, 2013.

\bibitem[BS07]{beigi2007complexity}
Salman Beigi and Peter~W. Shor.
\newblock {On the Complexity of Computing Zero-error and Holevo Capacity of
  Quantum Channels}.
\newblock arXiv:0709.2090, 2007.

\bibitem[BTW19]{boreland2019sandwich}
Gareth Boreland, Ivan~G. Todorov, and Andreas Winter.
\newblock {Sandwich Theorems and Capacity Bounds for Non-commutative Graphs}.
\newblock arXiv:1907.11504, 2019.

\bibitem[CCH11]{cubitt2011}
T.~S. Cubitt, J.~Chen, and A.~W. Harrow.
\newblock {Superactivation of the Asymptotic Zero-Error Classical Capacity of a
  Quantum Channel}.
\newblock {\em IEEE Transactions on Information Theory}, 57(12):8114--8126, Dec
  2011.

\bibitem[CLO15]{CLO15}
D.A. Cox, J.~Little, and D.~O'Shea.
\newblock {\em Ideals, Varieties, and Algorithms: An Introduction to
  Computational Algebraic Geometry and Commutative Algebra}.
\newblock Springer Publishing Company, Incorporated, 4th edition, 2015.

\bibitem[DSW13]{duan2013}
Runyao Duan, Simone Severini, and Andreas Winter.
\newblock {Zero-Error Communication via Quantum Channels, Noncommutative
  Graphs, and a Quantum Lov\'{a}sz Number}.
\newblock {\em IEEE Transactions on Information Theory}, 59(2):1164--1174, Feb
  2013.

\bibitem[Dua09]{duan2009super}
Runyao Duan.
\newblock {Super-activation of Zero-error Capacity of Noisy Quantum Channels}.
\newblock arXiv: 0906.2527, 2009.

\bibitem[GW90]{54907}
F.~{Guo} and Y.~{Watanabe}.
\newblock {On Graphs in which the Shannon Capacity is Unachievable by Finite
  Product}.
\newblock {\em IEEE Transactions on Information Theory}, 36(3):622--623, May
  1990.

\bibitem[Hae78]{Haemers1978}
Willem Haemers.
\newblock {An Upper Bound for the Shannon Capacity of a Graph}.
\newblock {\em Colloquia Mathematica Societatis J\'{a}nos Bolyai}, 25:267--272,
  1978.

\bibitem[Hae79]{haemers1979some}
Willem Haemers.
\newblock {On Some Problems of Lov{\'a}sz Concerning the Shannon Capacity of a
  Graph}.
\newblock {\em IEEE Transactions on Information Theory}, 25(2):231--232, 1979.

\bibitem[HPRS17]{Hogben2017}
Leslie Hogben, Kevin~F. Palmowski, David~E. Roberson, and Simone Severini.
\newblock {Orthogonal Representations, Projective Rank, and Fractional Minimum
  Positive Semidefinite Rank: Connections and New Directions}.
\newblock {\em Electronic Journal of Linear Algebra}, 32:98--115, 2017.

\bibitem[Kar72]{MR0378476}
Richard~M. Karp.
\newblock {Reducibility among Combinatorial Problems}.
\newblock In {\em Complexity of computer computations ({P}roc. {S}ympos., {IBM}
  {T}homas {J}. {W}atson {R}es. {C}enter, {Y}orktown {H}eights, {N}.{Y}.,
  1972)}, pages 85--103. Plenum, New York, 1972.

\bibitem[Lov79]{lovasz1979shannon}
L{\'a}szl{\'o} Lov{\'a}sz.
\newblock {On the Shannon Capacity of a Graph}.
\newblock {\em IEEE Transactions on Information Theory}, 25(1):1--7, 1979.

\bibitem[LPT18]{8355678}
R.~H. {Levene}, V.~I. {Paulsen}, and I.~G. {Todorov}.
\newblock {Complexity and Capacity Bounds for Quantum Channels}.
\newblock {\em IEEE Transactions on Information Theory}, 64(10):6917--6928, Oct
  2018.

\bibitem[LZ18]{Li2018quantum}
Yinan Li and Jeroen Zuiddam.
\newblock {Quantum Asymptotic Spectra of Graphs and Non-commutative Graphs, and
  Quantum Shannon Capacities}.
\newblock arXiv:1810.00744, 2018.

\bibitem[MR16]{manvcinska2016quantum}
Laura Man{\v{c}}inska and David~E. Roberson.
\newblock {Quantum Homomorphisms}.
\newblock {\em J.~Combin. Theory Ser.~B}, 118:228--267, 2016.

\bibitem[NC10]{Nielsen2010}
Michael~A. Nielsen and Isaac~L. Chuang.
\newblock {\em {Quantum Computation and Quantum Information}}.
\newblock Cambridge university press, 2010.

\bibitem[OP15]{ORTIZ2015128}
Carlos~M. Ortiz and Vern~I. Paulsen.
\newblock {Lovász Theta Type Norms and Operator Systems}.
\newblock {\em Linear Algebra and its Applications}, 477:128--147, 2015.

\bibitem[Pee96]{peeters1996orthogonal}
Ren{\'e} Peeters.
\newblock Orthogonal representations over finite fields and the chromatic
  number of graphs.
\newblock {\em Combinatorica}, 16(3):417--431, 1996.

\bibitem[PS19]{polak2018new}
Sven Polak and Alexander Schrijver.
\newblock {New lower bound on the Shannon capacity of C7 from circular graphs}.
\newblock {\em Information Processing Letters}, 143:37--40, 2019.

\bibitem[Sha56]{MR0089131}
Claude~E. Shannon.
\newblock {The Zero Error Capacity of a Noisy Channel}.
\newblock {\em Institute of Radio Engineers, Transactions on Information
  Theory}, IT-2:8--19, 1956.

\bibitem[Sta16]{stahlke2016}
Dan Stahlke.
\newblock {Quantum Zero-Error Source-Channel Coding and Non-Commutative Graph
  Theory}.
\newblock {\em IEEE Transactions on Information Theory}, 62(1):554--577, Jan
  2016.

\bibitem[Wat18]{watrous_2018}
John Watrous.
\newblock {\em {The Theory of Quantum Information}}.
\newblock Cambridge University Press, 2018.

\bibitem[WD18]{Wang2018}
X.~{Wang} and R.~{Duan}.
\newblock {Separation Between Quantum Lov\'{a}sz Number and
  Entanglement-Assisted Zero-Error Classical Capacity}.
\newblock {\em IEEE Transactions on Information Theory}, 64(3):1454--1460,
  March 2018.

\bibitem[Wea17]{weaver2017quantum}
Nik Weaver.
\newblock {A “Quantum” Ramsey Theorem for Operator Systems}.
\newblock {\em Proceedings of the American Mathematical Society},
  145(11):4595--4605, 2017.

\bibitem[Wea19]{weaver2018quantum}
Nik Weaver.
\newblock {The ``Quantum'' Turan Problem for Operator Systems}.
\newblock {\em Pacific Journal of Mathematics}, 301(1):335--349, 2019.

\end{thebibliography}

\end{document}